\documentclass[letterpaper,11pt]{article}
\usepackage{fullpage}
\usepackage{comment}
\usepackage[whole,autotilde]{bxcjkjatype}
\usepackage{amsthm,amsmath,amssymb}
\usepackage{algorithm,algpseudocode}
\usepackage[unicode]{hyperref}
\usepackage{booktabs}
\usepackage{authblk}
\usepackage{cite}
\usepackage{geometry}
\usepackage{enumitem}
\usepackage{multirow}
  
\theoremstyle{definition}
\newtheorem{theorem}{Theorem}[section]

\newtheorem{lemma}      {Lemma}[section]

\newtheorem{problem} {Problem}[section]
\newtheorem{corollary}   {Corollary}[section]

\numberwithin{equation}{section}
\numberwithin{figure}{section}
\numberwithin{table}{section}

\usepackage{autonum}

\newcommand{\E}{\mathbb{E}}
\renewcommand{\P}{\mathrm{Pr}}

\title{Stochastic Monotone Submodular Maximization with Queries}
\author[1]{Takanori Maehara\footnote{Email: takanori.maehara@riken.jp}}
\affil[1]{RIKEN Center for Advanced Intelligence Project}
\author[2,1]{Yutaro Yamaguchi\footnote{Email: yutaro\_yamaguchi@inf.kyushu-u.ac.jp}}
\affil[2]{Kyushu University}
\date{\empty}

\begin{document}
\maketitle
\thispagestyle{empty}

\begin{abstract}
We study a stochastic variant of monotone submodular maximization problem as follows.
We are given a monotone submodular function as an objective function and a feasible domain defined on a finite set, and our goal is to find a feasible solution that maximizes the objective function.
A special part of the problem is that each element in the finite set has a random hidden state, \emph{active} or \emph{inactive}, only the active elements contribute to the objective value, and we can conduct a query to an element to reveal its hidden state. 
The goal is to obtain a feasible solution having a large objective value by conducting a small number of queries.
This is the first attempt to consider nonlinear objective functions in such a stochastic model.

We prove that the problem admits a good query strategy if the feasible domain has a \emph{uniform exchange property}.
This result generalizes Blum et al.'s result on the unweighted matching problem and Behnezhad and Reyhani's result on the weighted matching problem in both objective function and feasible domain.
\end{abstract}
\clearpage
\thispagestyle{empty}
\tableofcontents
\clearpage
\setcounter{page}{1}

\section{Introduction}

\subsection{Background and Motivation}

The \emph{stochastic combinatorial optimization with queries} is the following type problem.
Let $E$ be a finite set, $f \colon 2^E \to \mathbb{R}_{\ge 0}$ be an objective function, $\mathcal{D} \subseteq 2^E$ be a feasible domain, and $p \in (0, 1)$ be a probability parameter.
At the beginning, nature selects a random subset $A \subseteq E$ such that $\P(e \in A) = p$ for all $e \in E$ independently.\footnote{All the results in this study can be generalized for the activation model of $\P(e \in A) \ge p$.} 
An element $e \in E$ is \emph{active} if $e \in A$ and \emph{inactive} otherwise.
We do not know whether $e$ is active or not in advance, but by conducting a \emph{query} to $e$, we can obtain this information.
Let $Q \subseteq E$ be a set of query targets. 
We say that $Q$ has an \emph{approximation factor of $c \in \mathbb{R}_{\geq 0}$} if 
\begin{align}
\label{eq:problem}
    \max_{X \in \mathcal{D}} f(X \cap A \cap Q) \ge c \max_{Z \in \mathcal{D}} f(Z \cap A)
\end{align}
holds (with high probability or in expectation).
The goal of the problem is to design a query strategy that conducts a small number of queries having a large approximation factor.
We evaluate not only the number of queries but also the \emph{degree of adaptivity} of a query strategy.
The degree of adaptivity is the number of ``rounds'' of the query strategy, where it may conduct multiple queries in each round.
A smaller degree of adaptivity is preferred because it corresponds to the number of adaptive decisions made for queries; in particular, a query strategy with degree of adaptivity one is called non-adaptive.

The above problem generalizes the \emph{stochastic matching problem} of Blum et al.~\cite{blum2015ignorance}, in which the objective function is the cardinality function, and the feasible domain is the set of matchings in a given graph.
They showed that, for any $\epsilon > 0$, there is a query strategy that conducts $1/p^{O(1/\epsilon)}$ queries per vertex (more precisely, it has the degree of adaptivity of $1/p^{O(1/\epsilon)}$, and in each round it conducts a query to a matching), and gives a $(1 - \epsilon)$-approximate solution in expectation.
Assadi et al.~\cite{assadi2016stochastic} considered the same problem and improved the number of queries to $O(1/p \epsilon)$ and the approximation guarantee to with high probability.
Behnezhad and Reyhani~\cite{behnezhad2018almost} considered the \emph{stochastic weighted matching} whose objective function is $f(X) = \sum_{e \in X} w_e$ for $w \colon E \to \mathbb{R}_{\ge 0}$, and obtained the same guarantee as that of Blum et al.~\cite{blum2015ignorance} for the unweighted problem.
There are sequential improvements on completely non-adaptive strategies for these problems \cite{assadi2017stochastic, assadi2019towards, behnezhad2019stochastic, behnezhad2020stoc, behnezhad2020weighted}.

Beyond the matching problems, the authors~\cite{maehara2019stochastic} considered the following research question:
\begin{problem}
\label{prob:question}
What class of problems admits efficient query strategy?
\end{problem}
To answer this question, the authors~\cite{maehara2019stochastic} considered a general problem, \emph{stochastic packing integer programming problem}, in which the objective function is $f(X) = \sum_{e \in X} c_e$ with $c_e = O(1)$ and the feasible domain is $\mathcal{D} = \{\, X \subseteq E : \sum_{e \in X} a_{i,e} \le b_i~(\forall i \in [n] = \{1, 2, \dots, n\}) \,\}$, where all the coefficients are nonnegative integers, and derived a sufficient condition of having a good query strategy, which is described in terms of the dual problem.
This result typically gives a query strategy that conducts $O(\text{poly}(1/p, 1/\epsilon) \log n)$ queries per constraint to obtain a $(1 - \epsilon) \alpha$-approximate solution with high probability, where $\alpha$ is the integrality gap of the problem.

In this study, we further explore this research question (Problem~\ref{prob:question}).
Specifically, we consider \emph{stochastic monotone submodular maximization with queries} problem, which is a stochastic combinatorial optimization with queries problem whose objective function is a monotone submodular function.

\subsection{Our Contribution}

There are two approaches in the literature of the stochastic combinatorial optimization with queries: the first one is the \emph{local-search based framework}, and the other is the \emph{duality-based approach}.
In our submodular objective case, we employ the local-search based framework because there is no existing duality theory for submodular maximization problems.

We show that an ``exchange property'' gives a sufficient condition for the existence of an efficient query strategy.

\subsubsection{First Attempt with Large Degree of Adaptivity}

If we ignore the degree of adaptivity, we can simply construct a query strategy via a local search.
To be precise, we here consider a {$k$-exchange system} because it has  an ``essence'' of the local search algorithm for submodular maximization problem.

A hereditary set system $\mathcal{D} \subseteq 2^E$ is a \emph{$k$-exchange system}~\cite{feldman2011improved} if for every $X, Y \in \mathcal{D}$, there exists a collection of subsets $\{ T_y \}_{y \in Y \setminus X} \subseteq 2^{X \setminus Y}$ such that
\begin{enumerate}[label=(\arabic*),labelindent=.5\parindent,leftmargin=*]
\item each $T_y$ has the cardinality at most $k$,
\item each $x \in X \setminus Y$ appears at most $k$ times in the collection, and
\item for every subset $S \subseteq Y \setminus X$, we have $X \cup S \setminus \bigcup_{y \in S} T_y \in \mathcal{D}$. 
\end{enumerate}
A feasible region of the $k$-set packing problem is always a $k$-exchange system.
This is also true for the $k$-matroid intersection problem and further the $k$-matroid parity problem when the matroids in question are strongly base orderable, but this is not the case in general. 

We first introduce a local search algorithm for a $k$-exchange system.
At $0$-th step, we set $X_0 = \emptyset$.
For each step $t$, we find $e \in E \setminus X_t$ and $T \subseteq X_t$ with $|T| \le k$ such that $X_t \cup \{e \} \setminus T$ is in $\mathcal{D}$ and has the largest objective value.
Then, we update the solution by $X_{t+1} \leftarrow X_t \cup \{e\} \setminus T$.

We can see that this local-search algorithm has an approximation factor of $(1 - \epsilon)/(k + 1)$ as follows.
Let $Y$ be the optimal solution.
Then, by the exchange property between $X_t$ and $Y$, there exists a set family $\{ T_e \}_{e \in Y \setminus X_t}$ that satisfies the above three conditions.
This set family satisfies the following inequality (we omit the proof; see Lemma~\ref{lem:expectation} for a generalization):
\begin{align}
    \sum_{e \in Y \setminus X_t} \left( f(X_t \cup \{e\} \setminus T_e) - f(X_t) \right) \ge f(Y) - (k + 1) f(X_t).
\end{align}
Thus, by taking the maximum summand, and by the definition of $X_{t+1}$, we obtain $f(X_{t+1}) - f(X_t) \ge (1/n) (f(Y) - (k + 1) f(X_t))$, where $n$ is the maximum cardinality of the solution. 
From this inequality, after $N$ iterations, we obtain the following inequality (we also omit the proof; see Lemma~\ref{lem:expectation}):
\begin{align}
    f(Y) - (k + 1) f(X_N) \le \left(1 - \frac{1}{n}\right)^N f(Y).
\end{align}
Thus, by choosing $N = n \log (1 / \epsilon)$, we obtain a solution with an approximation factor of $(1 - \epsilon)/(k + 1)$.

Now we convert this local search algorithm as a query strategy.
In each step of the local search, we conduct a query to the selected $e$.
If $e$ is active, then we perform the exchange; otherwise, we skip $e$ and continue to the next element.
This query strategy gives the same solution to the local search algorithm applied to the omniscient problem.
Hence, it has an approximation factor of $(1 - \epsilon)/(k + 1)$.
Moreover, it conducts linearly many queries in the solution size, i.e., $O(n \log (1/\epsilon)/p)$ queries, with high probability.

\subsubsection{Our Contribution: Uniform Exchange Map}

The only one issue of the above strategy is that it has a large degree of adaptivity because it conducts one query per each round.
To reduce the degree of adaptivity, we have to conduct multiple queries simultaneously in each round.
In terms of the local search, it corresponds to performing multiple augmentations simultaneously.
This leads us to a new structural property of set systems.

Again, we consider a $k$-exchange system. 
In the above analysis of the local search, we exchanged a current solution $X$ by a single element $e$ as $X \cup \{e\} \setminus T_e$.
However, in reality, the property (3) allows us to exchange any subset $S \subseteq Y$ simultaneously as $X \cup S \setminus \bigcup_{e \in S} T_e$.
This property indicates the following strategy: conduct a query to all $e \in Y \setminus X$ and observe the set of active elements $R$ in $Y \setminus X$;
then, exchange all $R$ simultaneously as $X \cup R \setminus \bigcup_{e \in R} T_e$.
We can prove that this strategy has a provable approximation factor with a small degree of adaptivity.

We generalize the above strategy to a general set system by introducing a new concept as follows.
For two feasible sets, $X, Y \in \mathcal{D}$, an \emph{exchange map between $X$ and $Y$} is a pair of (possibly random) functions $S_{X,Y} \colon 2^{Y \setminus X} \to 2^{Y \setminus X}$ and $T_{X,Y} \colon 2^{Y \setminus X} \to 2^{X \setminus Y}$ such that for any $R \subseteq Y \setminus X$, we have
\begin{enumerate}[label=(\arabic*),labelindent=.5\parindent,leftmargin=*]
\item $S_{X,Y}(R) \subseteq R$, and
\item $X \cup S_{X,Y}(R) \setminus T_{X,Y}(R) \in \mathcal{D}$.
\end{enumerate}
An exchange map $(S_{X,Y}, T_{X,Y})$ is $(\alpha, \beta)$-uniform if
$\P (y \in S_{X,Y}(R)) \ge \alpha$ for all $y \in Y \setminus X$ and $\P ( x \in T_{X,Y}(R) ) \le \beta$ for all $x \in X \setminus Y$, where $\P$ is the probability over $R \subseteq Y \setminus X$ such that $\P(y \in R) = p$ independently randomly. 
Note that $\alpha, \beta$ will depend on $p$.
We say that $\mathcal{D}$ \emph{admits an $(\alpha, \beta)$-uniform exchange map} if for all $X, Y \in \mathcal{D}$ there is an $(\alpha, \beta)$-uniform exchange map.
We refer to the ratio $\alpha / \beta$ as the \emph{uniformity} of the exchange map.

A typical example of a uniform exchange map comes from a $k$-exchange system.
From the collection of subsets $\{ T_e \}_{e \in Y \setminus X}$ in the property of the $k$-exchange system, we obtain the following exchange map:
\begin{align}
    S_{X,Y}(R) = R, \qquad
    T_{X,Y}(R) = \bigcup_{y \in R} T_y.
\end{align}
By the property (3) of the $k$-exchange system, this forms an exchange map.
Moreover, we can see that $\P(y \in S_{X,Y}(R)) = p$ and $\P(x \in T_{X,Y}(R)) \le p k$.
Therefore, it admits $(p, p k)$-uniform exchange map.

Next, we propose a query strategy.
As with \cite{maehara2019stochastic}, we introduce two problems.
A \emph{pessimistic problem} is the problem in which all the non-queried elements are supposed to be inactive, and an \emph{optimistic problem} is the problem in which all the non-queried elements are supposed to be active. 
Our algorithm iteratively computes an $\eta$-approximate solution to the optimistic problem (using any algorithm),
and conducts queries to the solution.
After sufficient iterations, it computes an $\eta$-approximate solution to the pessimistic problem.
The overall strategy is shown in Algorithm~\ref{alg:main}.
We remark that this type strategy has been employed commonly in many of the existing studies of stochastic combinatorial optimization with queries.

\begin{algorithm}[tb]
\caption{Query Strategy}
\label{alg:main}
\begin{algorithmic}[1]
\For{$t = 1, 2, \dots, N$}
\State{Compute an $\eta$-approximate solution $Y_t$ to the optimistic problem.}
\label{line:step2}
\State{Query all $e \in Y_t$}
\EndFor
\State{Output an $\eta$-approximate solution $X_N$ to the pessimistic problem.}
\label{line:step5}
\end{algorithmic}
\end{algorithm}

Our main lemmas are presented below.
One is for the linear objective case, and the other is for the submodular objective case; these are proved in a similar way.
These lemmas say that the algorithm gives a good solution if $\mathcal{D}$ admits a uniform exchange map, where the approximation factor depends on the uniformity of the exchange map.

\begin{lemma}[Linear Objective Case]
\label{lem:main-linear}
Suppose that $f \colon 2^E \to \mathbb{R}_{\geq 0}$ is a monotone linear function, $\mathcal{D}$ admits an $(\alpha, \beta)$-uniform exchange map, and there is an $\eta$-approximation algorithm to maximize a monotone linear function over $\mathcal{D}$.
Then, for any $\epsilon, \delta \in (0, 1)$, by setting
\[N = \frac{16 \log (1/\min\{\delta, \epsilon\})}{\alpha \min\{2, \max\{\alpha, \beta\}\} \eta \epsilon},\]
the output of Algorithm~\ref{alg:main} gives a $\left((1 - \epsilon) \alpha \eta / \max \{\alpha, \beta\}\right)$-approximate solution with probability at least $1 - \delta$.
\end{lemma}

\begin{lemma}[Submodular Objective Case]
\label{lem:main}
Suppose that $f \colon 2^E \to \mathbb{R}_{\geq 0}$ is a monotone submodular function, $\mathcal{D}$ admits an $(\alpha, \beta)$-uniform exchange map, and there is an $\eta$-approximation algorithm to maximize a monotone submodular function over $\mathcal{D}$.
Then, for any $\epsilon, \delta \in (0, 1)$, by setting
\[N = \frac{16 \log (1/\min\{\delta, \epsilon\})}{\alpha \min\{2, (\alpha + \beta)\} \eta \epsilon},\]
the output of Algorithm~\ref{alg:main} gives a $\left((1 - \epsilon) \alpha \eta / (\alpha + \beta)\right)$-approximate solution with probability at least $1 - \delta$.
\end{lemma}
It should be emphasized that the uniform exchange map is only used in the analysis (i.e., the algorithm does not explicitly use it). 
Thus, if $\mathcal{D}$ admits multiple uniform exchange maps, the performance of the algorithm is the maximum of them.

The proof of the above lemmas are not so complicated (see Section~\ref{sec:proof}).
The actual contribution is introducing the concept of uniform exchange map to separate the probabilistic argument required to the problem and the combinatorial argument about the feasible domain, which gives an answer to our research question (Question~\ref{prob:question}).

\medskip

We prove the existence of uniform exchange maps for $k$-exchange systems (Lemma~\ref{lem:k-exchange}), $k$-intersection systems (Lemma~\ref{lem:k-intersection}), and knapsack constraints (Lemmas~\ref{lem:knapsack-light} and \ref{lem:knapsack-heavy}). 
The uniformity of these exchange maps are summarized in Table~\ref{tab:uem}.
Consequently, we obtain polynomial-time query strategies for these constraints, which are summarized in Table~\ref{tab:main}.
Here, we employ the following approximation algorithms.
For the $k$-exchange systems, we employ Arkin and Hassin~\cite{arkin1998local}'s $1/(k-1)$-approximation 
and 
Berman~\cite{berman2000d}'s $2/(k+1)$-approximation for linear maximization, and 
Feldman et al.~\cite{feldman2011improved}'s $(1/k)$-approximation and 
Ward~\cite{ward2012k+}'s $2/(k+3)$-approximation for submodular maximization.
For the $k$-intersection systems, we employ
Lee, Sviridenko, and Vondr\'ak~\cite{lee2009submodular}'s $1/(k-1)$-approximation for linear maximization and $(1/k)$-approximation for submodular maximization.
For the knapsack constraints, 
we employ the classical dynamic programming for linear maximization and
Sviridenko~\cite{sviridenko2004note}'s $(1 - 1/e)$-approximation for submodular maximization.

We compare our result and the existing results.
Because the set of matchings forms a $2$-exchange system (with $\eta = 1$),
our result gives a strategy for the stochastic weighted matching problem that has the degree of adaptivity of $O(1)$ and the approximation factor of $1 - \epsilon$ with high probability.
This approximation factor is the same as Blum et al.~\cite{blum2015ignorance}'s result and Assadi et al.~\cite{assadi2016stochastic}'s result on the unweighted case, and as Beznezhad and Beyhani~\cite{behnezhad2018almost}'s result on the weighted case.
Our result has the stronger stochastic guarantee (that is, with high probability) than Blum et al.~\cite{blum2015ignorance} and Beznezhad and Beyhani~\cite{behnezhad2018almost}'s results (that is, in expectation).
However, our result requires an exponentially larger number of queries than Assadi et al.~\cite{assadi2016stochastic}'s result on the unweighted case.
Very recently, Beznezhad and Derakhshan~\cite{behnezhad2020weighted} provided a completely non-adaptive strategy for the weighted case that gives $(1 - \epsilon)$-approximation in expectation with $O(1)$ queries, where the dependency on $1/\epsilon$ and $1/p$ is tetration, $O(1) \uparrow\uparrow \mathrm{poly}(1/\epsilon, 1/p)$. 

Because the set of $k$-set packings forms a $k$-exchange system,
our result gives a strategy for the stochastic weighted $k$-set packing problem that has the degree of adaptivity of $O(1)$ and the approximation factor of 
\[\max\left\{(1 - \epsilon)/(k-1)^2,\, (2 - \epsilon)/(k^2 - 1)\right\}\]
with high probability.
While this approximation factor is weaker than Blum et al.~\cite{blum2015ignorance}'s result on unweighted $k$-set packing, which is $(2 - \epsilon)/k$, ours has a stronger stochastic guarantee than theirs.

On more general problems, our new result basically outperforms the authors' previous result~\cite{maehara2019stochastic} because our new result requires $O(1)$ queries, whereas the previous result requires $O(\log n)$ queries.
Also, our approach can be applied to non-packing-type constraint (e.g., matroid bases), whereas the previous result can only be applied to packing-type problems.

\begin{table}[tb]
\centering{\small
    \begin{tabular}{c|ccc}
    \toprule
         Constraint & $\alpha$ & $\beta$ \\ \midrule 
         $k$-Exchange System & $\displaystyle\frac{p^{1/\epsilon} \epsilon^2}{3}$ & $\displaystyle\frac{p^{1/\epsilon} \epsilon^2 (k - 1 + \epsilon)}{3}$  \\[3mm]
         $k$-Intersection System & $p$ & $p k$ \\[1mm]
         Knapsack Constraint$^*$ & $p$ & $p$ \\
         \bottomrule
    \end{tabular}
    \caption{Uniform exchange maps for several constraints. (*) In the knapsack constraint case, it is assumed that there are no heavy items and the uniform exchange map is $2$-relaxed; see Section~\ref{sec:knapsack}.}}
    \label{tab:uem}
\end{table}

\begin{table}[tb]
    \centering{\small
    \begin{tabular}{c|ccc}
    \toprule
         \multirow{2}{*}{Constraint} & \multirow{2}{*}{Degree of Adaptivity} & \multicolumn{2}{c}{Approximation Factor} \\ 
         & & Linear & Submodular \\ \midrule 
         $k$-Exchange System & $\displaystyle \frac{\log(1/\epsilon)}{kp^{O(1/\epsilon)}}$ & $\displaystyle \max \left\{ \frac{1-\epsilon}{(k-1)^2},\, \frac{2-\epsilon}{k^2-1} \right\}$ & $\displaystyle \max \left\{ \frac{1 - \epsilon}{k^2},\, \frac{2 - \epsilon}{k(k+3)} \right\}$ \\[4mm]
         $k$-Intersection System & $\displaystyle O\left(\frac{\log(1/\epsilon)}{\min\{p^2k, p\} \epsilon}\right)$ & $\displaystyle \frac{1 - \epsilon}{k (k-1)}$ & $\displaystyle \frac{1 - \epsilon}{k (k+1)}$ \\[4mm]
         Knapsack Constraint & $\displaystyle O\left(\frac{\log(1/\epsilon)}{p^2 \epsilon}\right)$ & $\displaystyle \frac{1 - \epsilon}{5}$ & $\displaystyle  \frac{1 - 1/e - \epsilon}{4 + 2 (1 - 1/e)}$ \\
         \bottomrule
    \end{tabular}
    \caption{Performance of our query strategies with polynomial-time approximation algorithms. The approximation guarantee is with probability at least $1 - \epsilon$. 
    }}
    \label{tab:main}
\end{table}

\subsection{Other Related Work}

A combinatorial optimization problem having uncertainty in its parameters is a fundamental problem in both theory and application, and is studied for a long time; see \cite{shapiro2009lectures}.

``Reducing the uncertainty by conducting queries'' is a relatively new approach to the problem.
Several query models have been proposed so far.
The \emph{stochastic probing} model~\cite{gupta2013stochastic} requires that if one queries an element $e$ and if $e$ is active, $e$ must be a part of the solution, i.e., $X = Q \cap A$.
The submodular objective case was considered in this setting~\cite{adamczyk2016submodular,adamczyk2015non}.
The \emph{price of information} problem~\cite{singla2018price} deals with the same problem has no such requirement and the objective function contains the cost term, $g(Q)$.
Chugg and Maehara~\cite{chugg2019submodular} studied the common generalization of the stochastic probing and the price of information model and provided a result for the submodular objective case.

These models measure the performance of algorithms by comparing with the optimal query strategy as in online optimization, whereas our model measures the performance by comparing with the optimal omniscient solution.
Thus, the techniques used in these studies are very different.

\section{Preliminaries}
\label{sec:preliminaries}

A function $f \colon 2^E \to \mathbb{R}$ is
\begin{itemize}[labelindent=.5\parindent,leftmargin=*]
\item \emph{normalized} if $f(\emptyset) = 0$.
\item \emph{monotone} if 
$
f(X) \le f(Y) 
$
for all $X, Y \subseteq E$ with $X \subseteq Y$.
\item \emph{submodular} if 
$
f(X) + f(Y) \ge f(X \cup Y) + f(X \cap Y)
$
for all $X, Y \subseteq E$.
\end{itemize}

We use the following probabilistic inequalities.

\begin{lemma}[Reverse Markov inequality]
\label{lem:reverse-markov}
Let $Z$ be a random variable such that $\E[Z] \ge a$ and $Z \le b$ for some $a, b \in \mathbb{R}_{\ge 0}$.
Then, 
\begin{align}
    \P(Z \ge (1/2) \E[Z]) \ge \frac{a}{2 b}.
\end{align}
\end{lemma}
\begin{proof}
We use the Markov inequality to $b - Z$, which is a nonnegative random variable: 
\begin{align}
    \P(Z < (1/2) \E[Z]) 
    &= \P(b - Z  > b - (1/2) \E[Z]) \\
    &\le \frac{b - \E[Z]}{b - (1/2) \E[Z]}
    = 1 - \frac{(1/2) \E[Z]}{b - (1/2) \E[Z]}.
\end{align}
Thus, 
\[
    \P(Z \ge (1/2) \E[Z]) 
    > \frac{(1/2) \E[Z]}{b - (1/2) \E[Z]} 
    \ge \frac{\E[Z]}{2 b} 
    \ge \frac{a}{2 b}. \qedhere
\]
\end{proof}

The following inequality directly follows from Chernoff bounds (see, e.g., \cite{boucheron2013concentration}).

\begin{lemma}[Tail Inequality for Binomial Distribution]
\label{lem:tail-binomial}
For any $N \in \mathbb{Z}_{\ge 0}$ and $q \in (0, 1)$, a random variable $X$ that follows a binomial distribution $\mathrm{Binomial}(4 N / q, q)$ satisfies
\begin{align}
    \P( X \le N ) \le \exp \left( - N \right).
\end{align}
\end{lemma}

\section{Proof of Main Lemmas}
\label{sec:proof}

The proofs of the main lemmas are similar in both the linear and submodular case.
For simplicity, we first give a proof for the linear case.
Then, we give a proof for the submodular case.

\subsection{Linear Objective Case}

We first evaluate the expected gain of the random uniform exchange as follows.
\begin{lemma}
\label{lem:expectation-linear}
Let $X, Y \in \mathcal{D}$ and $(S_{X,Y}, T_{X,Y})$ be an $(\alpha, \beta)$-uniform exchange map between $X$ and $Y$.
Then, for any monotone linear function $f \colon 2^E \to \mathbb{R}_{\ge 0}$, we have
\begin{align}
\label{eq:expectation}
    \E[ f(X \cup S_{X,Y}(R) \setminus T_{X,Y}(R)) - f(X) ]
    \ge \alpha f(Y) - \max\{\alpha, \beta\} f(X).
\end{align}
\end{lemma}
\begin{proof}
By the definition of $(\alpha, \beta)$-uniform exchange maps, we have
\begin{align}
    \E[ f(X \cup S_{X,Y}(R) \setminus T_{X,Y}(R)) - f(X) ]
    &= \E[ f(S_{X,Y}(R)) - f(T_{X,Y}(R)) ] \\
    &\ge \alpha f(Y \setminus X) - \beta f(X \setminus Y) \\
    &\ge \alpha f(Y) - \max\{\alpha, \beta\} f(X) . \qquad \qedhere
\end{align}
\end{proof}
Using this lemma, we obtain an ($\alpha / \beta$)-approximate solution with high probability.

\begin{lemma}
\label{lem:withhighprobability-linear}
Let $X, Y \in \mathcal{D}$, and $(S_{X,Y}, T_{X,Y})$ be an $(\alpha, \beta)$-uniform exchange map between $X$ and $Y$, and $f \colon 2^E \to \mathbb{R}_{\ge 0}$ be a monotone linear function.
If $f(Y) \ge \eta \max_{Y' \in \mathcal{D}} f(Y')$ for some $\eta \in (0, 1]$ and $(1 - \epsilon) \alpha f(Y) \ge \beta f(X)$ for some $\epsilon \in (0, 1)$, then
\begin{align}
\label{eq:wighhighprobability-linear}
    \P\left(f(X \cup S_{X,Y}(R) \setminus T_{X,Y}(R)) - f(X) \ge \frac{1}{2} \left( \alpha f(Y) - \max\{\alpha,\beta\} f(X) \right) \right) \ge \frac{\epsilon \alpha \eta}{2}.
\end{align}
\end{lemma}
\begin{proof}
Let $\Delta = f(X \cup S_{X,Y}(R) \setminus T_{X,Y}(R)) -  f(X)$ be a random variable in $R$ and let $\bar \Delta = (1/\eta) f(Y) - f(X)$.
Then, by Lemma~\ref{lem:expectation-linear} and the assumption of the lemma, we have $\E[ \Delta ] \ge \alpha f(Y) - \max \{\alpha, \beta\} f(X) \ge \epsilon \alpha f(Y)$.
Also, because $Y$ is an $\eta$-approximate solution, we have $\Delta \le \bar \Delta$.
Therefore, by the reverse Markov inequality (Lemma~\ref{lem:reverse-markov}), we have
\begin{align}
    \P( \Delta \ge (1/2) \E[ \Delta ] )
    \ge \frac{ \epsilon \alpha f(Y) }{(2/\eta) f(Y) - 2 f(X)} 
    \ge \frac{\epsilon \alpha \eta}{2}.
\end{align}
By expanding $\E[ \Delta ]$ using Lemma~\ref{lem:expectation-linear}, we obtain the lemma.
\end{proof}

\begin{proof}[Proof of Lemma~\ref{lem:main-linear}]
Let $\mathcal{F}_t$ be the filtration about the known active elements at $t$-th step.
Let $X_t$, $Y_t$, and $X^*$ be optimal solutions to the pessimistic, optimistic, and omniscient problems, respectively.
Note that all of these quantities (including $X^*$) are random variables depending on the activation of elements.
Let $\Delta_t = \alpha \eta f(X^*) - \max \{\alpha, \beta\} f(X_t)$.
Our goal is to bound the probability of $\Delta_N \leq \epsilon \Delta_0$, which implies that $X_N$ is a $\left((1 - \epsilon)\alpha\eta/\max\{\alpha, \beta\}\right)$-approximate solution.

If $(1 - \epsilon) \alpha \eta f(X^*) \ge \beta f(X_t)$ then, as $f(Y_t) \geq f(X^*)$, we have
$(1 - \epsilon) \alpha f(Y_t) \ge \beta f(X_t)$. 
Thus, by Lemma~\ref{lem:withhighprobability-linear},
\begin{align}
  \Delta_{t+1} - \Delta_t 
  &= -\max \{\alpha, \beta\} ( f(X_{t+1}) - f(X_t) ) \\
  &\le -\frac{\max\{\alpha,\beta\}}{2} \left( \alpha f(Y_{t}) - \max \{\alpha, \beta\} f(X_t) \right) \\
  &\le -\frac{\max\{\alpha,\beta\}}{2} \left( \alpha \eta f(X^*) - \max \{\alpha, \beta\} f(X_t) \right) \\
  &= -\frac{\max\{\alpha,\beta\}}{2} \Delta_t
\label{eq:deltagap-linear}
\end{align}
holds with probability at least $\epsilon \alpha \eta / 2$.
By using the following relation
\begin{align}
    \Delta_t - \epsilon \Delta_0 
    &= \alpha \eta f(X^*) - \max \{\alpha, \beta\} f(X_t) - \epsilon \alpha \eta f(X^*) \\
    &\le (1 - \epsilon) \alpha \eta f(X^*) - \beta f(X_t),
\end{align}
we can rewrite \eqref{eq:deltagap-linear} as 
\begin{align}
    \Delta_{t+1} \le \begin{cases}
        \displaystyle \left(1 - \frac{\max\{\alpha, \beta\}}{2} \right) \Delta_t, & \displaystyle \Delta_t > \epsilon \Delta_0 \text{ and with probability at least } \frac{\epsilon \alpha \eta}{2}, \\
        \Delta_t, & \text{otherwise}.
    \end{cases}
\end{align}
If the first event occurs at least $2 \log (1/\epsilon) / \max \{\alpha, \beta\}$ times, we obtain $\Delta_N \le \epsilon \Delta_0$.
Let $Z$ be a random variable that follows $\mathrm{Binomial}(N, \epsilon \alpha \eta / 2)$.
If we set $N$ as Lemma~\ref{lem:main-linear}, by Lemma~\ref{lem:tail-binomial}, 
\begin{align}
\P\left(Z \leq \frac{2\log (1/\epsilon)}{\max\{\alpha, \beta\}}\right) &\leq \P\left(Z \leq \frac{2\log (1/\min\{\delta, \epsilon\})}{\min\{2, \max\{\alpha, \beta\}\}}\right)\\
&\leq \exp\left( \frac{2\log (1/\min\{\delta, \epsilon\})}{\min\{2, \max\{\alpha, \beta\}\}}\right) \leq \delta,
\end{align}
and we are done. 
\end{proof}

\subsection{Submodular Objective Case}

The proof for the submodular case is basically the same.
We use the following lemmas as an alternative to the linearity of the objective function.

\begin{lemma}[{Probabilistic version of \cite[Lemma~1.1]{lee2009submodular}}]
\label{lem:covering}
Let $S \subseteq E$ be a random variable such that $\P(e \in S) \ge \alpha$ for all $e \in E$.
Then, for any normalized monotone submodular function $f \colon 2^E \to \mathbb{R}_{\ge 0}$,
\begin{align}
    \E[ f(S) ] \ge \alpha f(E)
\end{align}
holds.
\end{lemma}

\begin{proof}
Without loss of generality, we assume that $E = [m] = \{1, \dots, m\}$ for some positive integer $m$.
Then, for any monotone submodular function $f \colon 2^E \to \mathbb{R}_{\ge 0}$, we have
\begin{align}
    f(S)
    &= \sum_{i \in [m]} \left( f(S \cap [i]) - f(S \cap [i-1]) \right) \\
    &\ge \sum_{i \in [m]} \left( f((S \cap \{i\}) \cup [i-1]) - f([i-1]) \right) \\
    &= \sum_{i \in [m]} 1[ i \in S ] \left( f([i]) - f([i-1]) \right),
\end{align}
where $1[ i \in S ]$ is the indicator of the event $i \in S$.
By taking the expectation over $S$ and using the monotonicity of the function, i.e., $f([i]) - f([i-1]) \ge 0$, we obtain
\begin{align}
\label{eq:covering}
    \E[ f(S) ] 
    \ge \alpha \sum_{i \in [m]} \left( f([i]) - f([i-1]) \right)
    = \alpha f([m]),
\end{align}
which concludes the proof.
\end{proof}

\begin{lemma}[Probabilistic version of {\cite[Lemma~1.2]{lee2009submodular}}]
\label{lem:covering-complement}
Let $T \subseteq E$ be a random variable such that $\P(e \in T) \le \beta$ for all $e \in E$. 
Then, for any normalized monotone submodular function $f \colon 2^E \to \mathbb{R}_{\ge 0}$, 
\begin{align}
    \E[ f(E) - f(E \setminus T) ] \le \beta f(E)
\end{align}
holds.
\end{lemma}

\begin{proof}
Without loss of generality, we assume that $E = [m]$ for some positive integer $m$.
Then, for any monotone submodular function $f \colon 2^E \to \mathbb{R}_{\ge 0}$, 
\begin{align}
    f(E) - f(E \setminus T)
    &= \sum_{i \in [m]} \left( f(E \setminus (T \cap [i-1])) - f(E \setminus (T \cap [i])) \right) \\
    &\le \sum_{i \in [m]} \left( f(E \setminus [i-1]) - f(E \setminus [i-1] \setminus (T \cap \{ i \})) \right) \\
    &= \sum_{i \in [m]} 1[i \in T] \left( f(E \setminus [i-1]) - f(E \setminus [i]) \right).
\end{align}
where $1[ i \in T ]$ is the indicator of the event $i \in T$.
By taking the expectation over $T$ and using the monotonicity of the function, i.e., $f(E \setminus [i-1]) \ge f(E \setminus [i])$, we obtain
\begin{align}
\label{eq:packing}
    \E[ f(E) - f(E \setminus T) ] 
    \le \beta \sum_{i \in [m]} \left( f(E \setminus [i-1]) - f(E \setminus [i]) \right)
    = \beta f([m]),
\end{align}
which concludes the proof.
\end{proof}

Using these lemmas, we obtain the following lemma as a counterpart of Lemma~\ref{lem:expectation-linear}.

\begin{lemma}[Submodular version of Lemma~\ref{lem:expectation-linear}]
\label{lem:expectation}
Let $X, Y \in \mathcal{D}$ and $(S_{X,Y}, T_{X,Y})$ be an $(\alpha, \beta)$-uniform exchange map between $X$ and $Y$.
Then, for any normalized monotone submodular function $f \colon 2^E \to \mathbb{R}_{\ge 0}$, we have
\begin{align}
\label{eq:expectation-submod}
    \E[ f(X \cup S_{X,Y}(R) \setminus T_{X,Y}(R)) - f(X) ]
    \ge \alpha f(X \cup Y) - (\alpha + \beta) f(X).
\end{align}
\end{lemma}
\begin{proof}
By submodularity of $f$, for any $R$, 
\begin{align}
    &\ f(X \cup S_{X,Y}(R) \setminus T_{X,Y}(R)) - f(X)\\
    \ge&\
    \underset{(1)}{\underline{f(X \cup S_{X,Y}(R)) - f(X)}} + \underset{(2)}{\underline{f(X \setminus T_{X,Y}(R)) - f(X)}}.
\end{align}
Then, we take the expectation over $R$.
(1) is lower-bounded by $\alpha f(X \cup Y) - \alpha f(X)$ by Lemma~\ref{lem:covering} applied to the function $2^{Y \setminus X} \ni S \mapsto f(X \cup S) - f(X)$, and (2) is lower-bounded by $-\beta f(X)$ by Lemma~\ref{lem:covering} applied to the function $2^X \ni S \mapsto f(S)$.
Thus, we obtain \eqref{eq:expectation-submod}.
\end{proof}

\begin{proof}[Proof of Lemma~\ref{lem:main}]
The proof is the same as proof of Lemma~\ref{lem:main-linear} where we use \mbox{Lemma~\ref{lem:expectation}} instead of Lemma~\ref{lem:expectation-linear}.
\end{proof}

\section{Examples of Uniform Exchange Maps}

In this section, we present several examples of uniform exchange maps.

\subsection{Exchange System}

As we see in Introduction, a $k$-exchange system admits a $(p, p k)$-uniform exchange map.
Here, we prove the existence of an exchange map with a better uniformity, which leads to better approximation ratios than $1/k$ in the linear case and $\eta/(k + 1)$ in the submodular case.

\begin{lemma}
\label{lem:k-exchange}
For any $h \in \mathbb{Z}_{\ge 1}$, any $k$-exchange system $\mathcal{D}$ admits an $(\alpha_h, \beta_h)$-uniform exchange map, where 
\[\alpha_h = \frac{p^h}{h}, \qquad \beta_h = \frac{p^h}{h}\cdot\left(k - 1 + \frac{1}{h}\right).\]
\end{lemma}

\begin{corollary}
For the monotone linear (resp., submodular) maximization problem on $k$-exchange systems, for any $\epsilon > 0$, Algorithm~\ref{alg:main} with
\[N = \frac{\log(1/\min\{\delta, \epsilon\})}{kp^{\Omega(1/\epsilon)}}\]
gives a solution whose approximation factor is $(1 - \epsilon)\eta/(k-1)$ (resp., $(1 - \epsilon) \eta / k$) with probability at least $1 - \delta$.
\end{corollary}
This result generalizes Blum et al.~\cite{blum2015ignorance}'s and Behnezhad and Reyhani~\cite{behnezhad2018almost}'s results on stochastic unweighted and weighted matching problems because the set of matchings forms a $2$-exchange system.
Our proof also generalizes their proofs using the technique in \cite{feldman2011improved} for a local search on $k$-exchange systems. 
We use the following lemma.
\begin{lemma}[\!\!{\cite[Theorem~5]{feldman2011improved}}]
\label{lem:feldman2011improved-theorem5}
Let $G$ be an undirected graph whose maximum degree is at most $k \in \mathbb{Z}_{\ge 2}$.
  Then, for every $h \in \mathbb{Z}_{\ge 1}$, there exists a multiset $\mathcal{P}(G, k, h)$ of simple paths in $G$ and a labeling $\ell \colon V \times \mathcal{P}(G, k, h) \to \{0 \} \cup [h]$ such that the following properties hold.
\begin{enumerate}[label=\arabic*.,labelindent=.5\parindent,leftmargin=*]
\item 
  For every $P \in \mathcal{P}(G, k, h)$, the labeling $\ell$ of the vertices in $P$ is consecutive and increasing with labels from $[h]$.
Vertices not in $P$ receive label $0$.
\item 
  For every $P \in \mathcal{P}(G, k, h)$ and $v \in P$, if $\mathrm{deg}_G(v) = k$ and $\ell(v, P) \not \in \{1, h \}$, then at least two of the neighbors of $v$ are in $P$.
\item 
  For every $v \in V$ and label $i \in [h]$, there are $n(k, h) = k (k - 1)^{h - 2}$ paths $P \in \mathcal{P}(G, k, h)$ for which $\ell(v, P) = i$.
\end{enumerate}
\end{lemma}

\begin{proof}[Proof of Lemma~\ref{lem:k-exchange}]
We fix $X, Y \in \mathcal{D}$ and construct an exchange map as follows.
Let $\{ T_y \}_{y \in Y \setminus X}$ be the subsets in the definition of the $k$-exchange system.
Let $\mathcal{G} = (\mathcal{V}, \mathcal{E})$ be a bipartite graph where $\mathcal{V} = (X \setminus Y) \cup (Y \setminus X)$ and $\mathcal{E} = \{\, (x, y) \in (X \setminus Y) \times (Y \setminus X) : x \in T_y \,\}$.
By the definition of $k$-exchange system, the maximum degree of $\mathcal{G}$ is at most $k$.
Hence, by Lemma~\ref{lem:feldman2011improved-theorem5}, we obtain paths $\{ P_1, \dots, P_M \} = \mathcal{P}(\mathcal{G}, k, 2 h)$, where $M \geq |\mathcal{V}| \cdot n(k, 2h)$ by the first and third properties.
Let $S_i = P_i \cap (Y \setminus X)$ and $T_i = \bigcup_{y \in S_i} T_y$ for each $i \in [M]$.
Then, we have $|S_i| \le h$ by the first property, and $|T_i| \le 1 + h(k-1)$.

We consider the intersection graph of $\{ S_i \}_{i \in [M]}$: the vertices are $\{ S_i \}_{i \in [M]}$ and there is an edge $(S_i, S_j)$ if and only if $S_i \cap S_j \neq \emptyset$.
By the third property of the multiset $\mathcal{P}(\mathcal{G}, k, 2h)$, the maximum degree of this graph is at most $2h^2(n(k, 2 h) - 1)$.
Therefore, it admits a $2 h^2 n(k, 2 h)$-vertex coloring (see, e.g., Proposition~5.2.2 in \cite{diestel2017graph}); we choose any such coloring. 

Now we define $S_{X,Y}$ and $T_{X,Y}$.
We first draw a color class uniformly at random.
Then, we select each $S_i$ in the drawn color class if all $y \in S_i$ are active and with probability $p^{h - |S_i|}$. 
Here, the additional probability makes the selection probability uniform.
We define $S_{X,Y}$ and $T_{X,Y}$ by
\begin{align}
    S_{X,Y}(R) = \bigcup_{i\colon S_i \text{ is selected}} S_i, \qquad
    T_{X,Y}(R) = \bigcup_{i\colon S_i \text{ is selected}} T_i. 
\end{align}
We check these form an $(\alpha_h, \beta_h)$-uniform exchange map.
By the construction, for any $R \subseteq Y \setminus X$, we have $S_{X,Y}(R) \subseteq R$.
Also, by the definition of the $k$-exchange system, $X \cup S_{X,Y}(R) \setminus T_{X,Y}(R) \in \mathcal{D}$ holds.

The event ``$y \in S_{X,Y}(R)$'' occurs when the drawn color class contains $S_i$ with $y \in S_i$ and $S_i$ is selected. 
By the third property in Lemma~\ref{lem:feldman2011improved-theorem5}, each $y \in Y \setminus X$ is contained in exactly $2hn(k, 2h)$ sets $S_i$, which are all intersecting at $y$ in the intersection graph.
Thus, we obtain
\[\P\left(y \in S_{X,Y}(R)\right) = \frac{2hn(k, 2 h)}{2 h^2 n(k, 2 h)} \cdot p^h = \frac{p^h}{h} = \alpha_h.\]

The event ``$x \in T_{X,Y}(R)$'' occurs when the drawn color class contains $S_i$ with $x \in T_i$ and $S_i$ is selected.
Since $x$ has at most $k$ neighbors in $\mathcal{G}$, there are at most $k\cdot2hn(k, 2h)$ sets $S_i$ with $x \in T_i$.
Moreover, by the second property in Lemma~\ref{lem:feldman2011improved-theorem5}, if $x$ has exactly $k$ neighbors, then at least $2(h - 1)n(k, 2h)$ among such $S_i$ contain at least two neighbors of $x$; hence, they are doubly counted. 
Thus, we obtain
\[\P\left(x \in T_{X,Y}(R)\right) \leq \frac{2(1 + h(k - 1))n(k, 2 h)}{2 h^2 n(k, 2 h)} \cdot p^h = \frac{p^h}{h}\cdot \left(k - 1 + \frac{1}{h}\right) = \beta_h. \qedhere\]
\end{proof}

\subsection{Matroid and Matroid Intersection}

A family of subsets $\mathcal{I} \subseteq 2^E$ is an \emph{independent set family} of a matroid if
\begin{enumerate}[label=(\arabic*),labelindent=.5\parindent,leftmargin=*]
\item $\emptyset \in \mathcal{I}$,
\item $Y \in \mathcal{I}$ implies $X \in \mathcal{I}$ for all $X \subseteq Y$, and
\item if $X, Y \in \mathcal{I}$ and $|X| < |Y|$ then there exists $e \in Y \setminus X$ such that $X \cup \{e\} \in \mathcal{I}$.
\end{enumerate}
An element $I \in \mathcal{I}$ is called an \emph{independent set}.
A maximal element $B \in \mathcal{I}$ is called a \emph{base}. 
The set $\mathcal{B}$ of all the bases are called the \emph{base family} of the matroid.

Both the independence family and the base family satisfy the following property called the generalized Rota-exchange property.
\begin{theorem}[\!\!{\cite[Lemma~2.6]{lee2009submodular}}]\
Let $X, Y \in \mathcal{B}$. 
For any subsets $A_1, \dots, A_n \subseteq Y \setminus X$ that cover each $y \in Y \setminus X$ exactly $q$ times, there exists subsets $B_1, \dots, B_n \subseteq X \setminus Y$ such that $X \cup A_i \setminus B_i \in \mathcal{B}$ for all $i \in [n]$ and each $x \in X \setminus Y$ is covered exactly $q$ times.
\end{theorem}
\begin{theorem}[\!\!{\cite[Lemma~2.7]{lee2009submodular}}]\
Let $X, Y \in \mathcal{I}$. 
For any subsets $A_1, \dots, A_n \subseteq Y \setminus X$ that cover each $y \in Y \setminus X$ at most $q$ times, there exists subsets $B_1, \dots, B_n \subseteq X \setminus Y$ such that $X \cup A_i \setminus B_i \in \mathcal{I}$ for all $i \in [n]$ and each $x \in X \setminus Y$ is covered at most $q$ times.
\end{theorem}

These theorems immediately give the uniform exchange map as follows.
\begin{lemma}
\label{lem:matroid}
An independent set family and a base family of a matroid admit $(p, p)$-uniform exchange maps.
\end{lemma}
\begin{proof}
We apply the generalized Rota-exchange property to the family $2^{Y \setminus X}$ to obtain a family of subsets $\{ B_R \}_{R \subseteq Y \setminus X}$. 
We define $S_{X,Y}(R) = R$ and $T_{X,Y}(R) = B_R$.
Then, the generalized Rota-exchange property guarantees that $\P(y \in S_{X,Y}(R)) = p$ ($ = \alpha$) for all $y \in Y \setminus X$ and $\P(x \in T_{X,Y}(R)) \le p$ ($= \beta$) for all $x \in X \setminus Y$.
\end{proof}

We then consider the intersection of $k$ matroids.
The exchange map has the following composition property.
\begin{lemma}
\label{lem:composition}
Let $\mathcal{D}_1, \dots, \mathcal{D}_n \subseteq 2^E$ be set families each of which admits $(\alpha,\beta_i)$-uniform exchange maps $(S_{X,Y}, T_{i,X,Y})$ with the same $S_{X, Y}$.
Then, $\mathcal{D}_1 \cap \dots \cap \mathcal{D}_n$ admits an $(\alpha, \beta_1 + \dots + \beta_n)$-uniform exchange map.
\end{lemma}
\begin{proof}
We can define $T_{X,Y}(R) = T_{1,X,Y}(R) \cup \dots \cup T_{n,X,Y}(R)$.
\end{proof}
Using this lemma, we immediately obtain the following result.
\begin{lemma}
\label{lem:k-intersection}
A $k$-intersection system admits a $(p, pk)$-uniform exchange map.
\end{lemma}
\begin{corollary}
For the monotone linear (resp., submodular) maximization problem on a $k$-intersection system, for any $\epsilon > 0$, Algorithm~\ref{alg:main} with 
\[N = \Omega\left(\frac{\log(1/\min\{\delta, \epsilon\})}{\min\{p^2k, p\}\epsilon}\right)\]
gives a solution whose approximation factor is $(1 - \epsilon)\eta/(k + 1)$ (resp., $(1 - \epsilon)\eta/k$) with probability at least $1 - \delta$.
\end{corollary}

\noindent
\emph{Remark.}~
We tried to obtain an exchange map with better uniformity using the local search technique in \cite{lee2009submodular}, but it has not succeeded.

\subsection{Knapsack Constraint}
\label{sec:knapsack}

We refer as a \emph{knapsack constraint} to a family $\mathcal{D} = \{\, X \subseteq E : \sum_{x \in X} c_x \le 1 \,\}$ for some positive numbers $c_e \in \mathbb{R}_{> 0}$ ($e \in E$).
In general, a knapsack constraint may not have a uniform exchange map with a small uniformity because adding one ``heavy'' item may require to remove almost all the items.
To handle this situation, we handle heavy items and light items separately; then combine these results.
Note that a similar technique can be found in the literature of prophet inequality~\cite{dutting2017prophet}.

We first generalize our definition as follows.
For a positive integer $\gamma$, we define $\gamma \mathcal{D} := \{\, X_1 \cup \dots \cup X_\gamma : X_1, \dots, X_\gamma \in \mathcal{D} \,\}$.
For two feasible sets $X, Y \in \mathcal{D}$, a \emph{$\gamma$-relaxed exchange map} is a pair of functions $(S_{X,Y}, T_{X,Y})$ such that for any $R \subseteq Y \setminus X$, we have
\begin{enumerate}[label=(\arabic*),labelindent=.5\parindent,leftmargin=*]
\item $S_{X,Y}(R) \subseteq R$, and
\item $X \cup S_{X,Y}(R) \setminus T_{X,Y}(R) \in \gamma \mathcal{D}$.
\end{enumerate}
The $(\alpha, \beta)$-uniformity is defined similarly to the exchange map.
A feasible domain $\mathcal{D}$ admits $\gamma$-relaxed $(\alpha, \beta)$-exchange map if for any $X, Y \in \mathcal{D}$, there exists $\gamma$-relaxed $(\alpha, \beta)$-exchange map.

Lemmas~\ref{lem:main-linear} and \ref{lem:main} can be generalized as follows.
\begin{lemma}[Generalization of Lemmas~\ref{lem:main-linear} and~\ref{lem:main}]
\label{lem:main-generalized}
Suppose that $\mathcal{D}$ admits a $\gamma$-relaxed $(\alpha, \beta)$-uniform exchange map.
If the objective function is monotone linear (resp., submodular), 
for any $\epsilon, \delta \in (0, 1)$, by setting $N$ as with Lemma~\ref{lem:main-linear} (resp., Lemma~\ref{lem:main}), 
the output of Algorithm~\ref{alg:main} gives a $(1 - \epsilon) \alpha \eta / \gamma \max \{\alpha, \beta\}$ (resp., $(1 - \epsilon) \alpha \eta / \gamma (\alpha + \beta)$)-approximate solution with probability at least $1 - \delta$.
\end{lemma}
\begin{proof}
The proof is almost the same as that of Lemmas~\ref{lem:main-linear} and \ref{lem:main}, where we use the $\gamma$-relaxed condition to guarantee the optimal value at $(t+1)$-th step is at most $\gamma$ times the exchanged solution.
\end{proof}

\paragraph{Special Case 1: No Heavy Items}
We say that an item $e \in E$ is \emph{heavy} if $c_e > 1/3$.
We first consider the case that there is no heavy item.

\begin{lemma}
Suppose that there is no item with size greater than $1/3$.
Then, $\mathcal{D}$ admits a $2$-relaxed $(p, p)$-uniform exchange map.
\end{lemma}
\begin{proof}
Without loss of generality, we consider $X \cap Y = \emptyset$.
We set $S_{X,Y}(R) = R$, which gives $\P(y \in S_{X,Y}(R)) = p$.
We define $T_{X,Y}(R)$ as follows.
We write a unit length circle and pack each element $x \in X$ as an arc $\gamma_x$ of length $c_x$.
Then, we select an arc $\gamma$ of length $\sum_{r \in R} c_r$ uniformly at random.
We remove each item $x$ with probability $|\gamma_x \cap \gamma| / c_x$.
Due to the randomness of $\gamma$, it gives $\P(x \in T_{X,Y}(R)) = p \sum_{y \in Y \setminus X} c_y \leq p$.
Because the boundary of $\gamma$ overlaps at most two subsets, the exchanged solution has the capacity at most $5/3$, which is decomposed to two sets of capacity $1$ ($\because$ place the items into the interval of length $5/3$ and partition by length $1$).
Thus, this procedure gives $2$-relaxed $(p, p)$-uniform exchange map.
\end{proof}

This lemma immediately gives the following result. 
Note that we can use the FPTAS for linear objective case.
\begin{lemma}
\label{lem:knapsack-light}
For the monotone linear (resp., submodular) maximization problem on a knapsack constraint without items of size greater than $1/3$, there is a query strategy that conducts $O(\log (1/\min\{\delta, \epsilon\})/p^2 \epsilon)$ queries and has an approximation factor of $(1 - \epsilon)/3$ (resp., $(1 - \epsilon) \eta/4)$) with probability at least $1 - \delta$.
\end{lemma}

\paragraph{Special Case 2: Only Heavy Items}
If all the items have size greater than $1/3$, the cardinality of a solution is at most $2$.
In general, if the cardinality of a solution is at most $k$, we obtain a query strategy with a good approximation guarantee as follows.

\begin{lemma}
\label{lem:knapsack-heavy}
For both the monotone linear and submodular maximization problems on a constraint $\mathcal{D}$ such that $|X| \le k$ for all $X \in \mathcal{D}$, there is a query strategy that conducts $O(\log (1 / \min\{\delta, \epsilon\}) / p^2 \epsilon)$ queries and has an approximation factor of $(1 - \epsilon) p / (1 - (1-p)^k)$ ($\ge (1 - \epsilon)/k$) with probability at least $1 - \delta$.
\end{lemma}
\begin{proof}
We construct a uniform exchange map as follows.
We define $S_{X,Y}(R) = R$. Then $\P(y \in S_{X,Y}(R)) = p$.
We define $T_{X,Y}(R)$ by 
\begin{align}
    T_{X,Y}(R) = \begin{cases}
        \emptyset, & R = \emptyset, \\
        X, & \text{otherwise}.
    \end{cases}
\end{align}
Then, $\P(x \in T_{X,Y}(R)) = 1 - (1 - p)^k$.
Using this map, we can see that
\begin{align}
    f(X \cup S_{X,Y}(R) \setminus T_{X,Y}(R)) - f(X) =
    \begin{cases}
        0, & R = \emptyset, \\
        f(R) - f(X), & \text{otherwise}.
    \end{cases}
\end{align}
Thus, by taking the expectation, we obtain
\begin{align}
    \E[ f(X \cup S_{X,Y}(R) \setminus T_{X,Y}(R)) ]
    &= \E[ f(R) ] - (1 - (1-p)^k) f(X) \\
    &\ge p f(Y) - (1 - (1-p)^k) f(X).
\end{align}
Using this inequality instead of \eqref{eq:expectation}, we obtain the desired result.
Note that, in this case, in each iteration, we can use the exact brute-force algorithm, which runs in a polynomial-time.
\end{proof}

\paragraph{General Case}~
By combining the query strategies in Lemmas~\ref{lem:knapsack-light} and \ref{lem:knapsack-heavy}, we obtain the following theorem.

\begin{theorem}
The knapsack problem whose objective function is monotone linear (resp., submodular) admits a query strategy that conducts $O(\log (1/\min\{\delta, \epsilon\}) / p^2 \epsilon)$ queries and has an approximation factor of $(1 - \epsilon)/5$ (resp., $(1 - \epsilon) \eta/(4 + 2 \eta)$) with probability at least $1 - \delta$.
\end{theorem}
\begin{proof}
We only prove the linear case. The submodular case is proved similarly.
We run the strategies of two special cases simultaneously.
We evaluate the performance of this strategy.

Let $X^*$ be an optimal solution, $X^{* (l)}$ be an optimal solution that consists of items of size at most $1/3$, and $X^{* (h)}$ be an optimal solution that consists of items of size greater than $1/3$, respectively, to the omniscient problem.
Also, let $X_N$ be an optimal solution, $X_N^{(l)}$ be the subset of $X_N$ that consists of items of size at most $1/3$, and $X_N^{(h)}$ be the subset of $X_N$ that consists of items of size greater than $1/3$, respectively, to the pessimistic problem.
Then, 
\begin{align}
    f(X_N \cap A \cap Q) &\ge f(X_N^{(l)} \cap A \cap Q) \ge \frac{1 - \epsilon}{3}f(X^{* (l)} \cap A), \\
    f(X_N \cap A \cap Q) &\ge f(X_N^{(h)} \cap A \cap Q) \ge \frac{1 - \epsilon}{2}f(X^{* (h)} \cap A),
\end{align}
simultaneously with probability at least $1 - 2 \delta$.
Therefore, we obtain
\begin{align}
    f(X_N \cap A \cap Q) 
    &\ge \frac{1 - \epsilon}{5} \left( f(X^{* (l)} \cap A) + f(X^{* (h)} \cap A) \right) \\
    &\ge \frac{1 - \epsilon}{5} f(X^* \cap A).
\end{align}
\end{proof}

\section*{Acknowledgments}
TM is supported by JSPS KAKENHI Grant Number 19K20219.
Most of this work was done when YY was with Osaka University.

\bibliographystyle{plain}
\bibliography{main}

\end{document}